\documentclass[letterpaper, 10 pt, conference]{ieeeconf}  

\IEEEoverridecommandlockouts                              

\usepackage{graphics} 
\usepackage{mathptmx} 
\usepackage{amsmath} 
\usepackage{amssymb}  
\usepackage[noadjust]{cite}

\title{\LARGE \bf
Disease Spread over Randomly Switched Large-Scale Networks
}

\author{Masaki Ogura$^{1}$ and Victor M.~Preciado$^{2}$
\thanks{$^{1}$M.~Ogura is with the Department of Mathematics and Statistics, Texas Tech University, TX 79409, USA. {\tt\small msk.ogura@gmail.com}, {\tt\small masaki.ogura@ttu.edu}}%
\thanks{$^{2}$V.M.~Preciado is with the Department of Electrical and Systems Engineering at the University of Pennsylvania, Philadelphia, PA 19104, USA. {\tt\small preciado@seas.upenn.edu}}%
}

\newtheorem{theorem}{Theorem}

\newtheorem{proposition}[theorem]{Proposition}
\newtheorem{lemma}[theorem]{Lemma}

\newtheorem{definition}[theorem]{Definition}
\newtheorem{assumption}[theorem]{Assumption}
\newtheorem{example}[theorem]{Example}

\newtheorem{remark}[theorem]{Remark}

\numberwithin{theorem}{section}

\newenvironment{proofof}[1]{
\renewcommand\proof{\noindent\hspace{2em}{\itshape Proof of #1: }}%
\begin{proof}}{\end{proof}
}

\usepackage{graphicx}
\usepackage{subcaption}

\usepackage{mathtools}
\mathtoolsset{showonlyrefs=true}
\mathtoolsset{centercolon}

\newcommand{\norm}[1]{\lVert #1 \rVert}

\newcommand{\abs}[1]{\lvert #1 \rvert}

\DeclareMathOperator{\diag}{diag}
\DeclareMathOperator{\Var}{Var}

\DeclareSymbolFont{bbold}{U}{bbold}{m}{n}
\DeclareSymbolFontAlphabet{\mathbbold}{bbold}
\newcommand{\onev}{\mathbbold{1}}

\newcommand\hyperrefopt{colorlinks=true,linkcolor={black},citecolor={black},urlcolor={black},pdfauthor={Masaki Ogura}}
\usepackage[\hyperrefopt]{hyperref}

\newcommand{\afterequation}{\vskip 0pt}

\begin{document}

\maketitle
\thispagestyle{empty}
\pagestyle{empty}

\newcommand{\nnminumoneot}{\cramped{2^{n{\scriptscriptstyle(}n-1{\scriptscriptstyle)/}2}}}

\begin{abstract}
In this paper we study disease spread over a randomly switched
network, which is modeled by a stochastic switched differential
equation based on the so called $N$\nobreakdash-intertwined model for
disease spread over static networks. Assuming that all the edges of
the network are independently switched, we present sufficient
conditions for the convergence of infection probability to zero.
Though the stability theory for switched linear systems can
naively derive a necessary and sufficient condition for the
convergence, the condition cannot be used for large-scale networks
because, for a network with $n$ agents, it requires computing the
maximum real eigenvalue of a matrix of size exponential in $n$. On the
other hand, our conditions that are based also on the spectral theory
of random matrices can be checked by computing the maximum real
eigenvalue of a matrix of size exactly $n$.
\end{abstract}


\section[Introduction]{Introduction}

Network epidemiology, a branch of mathematical epidemiology which aims
to mathematically understand disease spread over networks, has been
attracting ever-growing attention~\cite{Pastor-Satorras2014a}. One of
the reasons comes from the finding that the heterogeneous structure of
real world networks, such as an uneven distribution of the number of
neighbors of individuals, cannot be ignored to perform an accurate
prediction of epidemics~\cite{Chakrabarti2008}. Another reason comes
from recent advancements of the technologies for collecting a massive
amount of data from real human interactions, which enable researchers
to build accurate mathematical models of human networks (see
\cite{Vespignani2009} and references therein). One of the recent
achievements of network epidemiology, for example, is the effective
prediction of the pandemic of H1N1 influenza on
2009~\cite{Balcan2009}.

In the last decade, the major emphasis of network epidemiology has
been on the disease spread over {\it static}
networks~\cite{Pastor-Satorras2014a}. This time-invariance assumption
enables us to model disease spread as a time-homogeneous Markov
process, for which an effective approximation by a
constant-coefficient differential equation, called an $N$-intertwined
model, is available~\cite{Chakrabarti2008,VanMieghem2009a}. For
example it is known~\cite{Chakrabarti2008,VanMieghem2009a} that the
evolution of a certain epidemic model over a static network is
characterized by only the maximum real part of the eigenvalues of the
adjacency matrix of the network.

However, recently, several simulation studies have reported that the
{\it dynamical} nature of human interactions can greatly affect the
way epidemics evolve and that static approximations of such dynamical
networks may result in inaccurate conclusions~\cite{Masuda2013}. For
example, Vazquez et al.~\cite{Vazquez2007} observe that the speed of
disease spread can be substantially slowed down on randomly switched
networks than on static ones. The result by Volz and
Meyers~\cite{Volz2009} suggests that the frequency of the change of
the topology of network structure is an important factor for the
emergence of endemic. In \cite{Sun2014}, it is numerically confirmed
that so-called memory effect of networks can slow down disease spread.
Despite those simulation results, rigorous understanding of disease
spread over dynamical networks based on mathematical analysis is very
limited.

This paper gives an analysis of disease spread over randomly switched
networks with the stability theory for switched linear systems called
Markov jump linear systems~\cite{Costa2013} and the spectral theory of
random graphs~(see, e.g., \cite{Chung2003a}). We
provide sufficient conditions for epidemics to eventually die out over
randomly switched networks having independent and Markovian edges.
This paper can be considered as a stochastic counter-part of the
recently submitted paper~\cite{Rami2014}, where the authors study
epidemics over deterministically switched networks.

One of the major difficulties of the problem is in the computational
cost. The stability theory for Markov jump linear systems in fact
gives a necessary and sufficient condition for epidemics eventually
dying out as we will see later
(Proposition~\ref{prop:meanstbl:Sigma}). {Unfortunately, checking
the condition requires finding the maximum real eigenvalue of a matrix
with size $\cramped{2^{n(n-1)/2}}$, where $n$ denotes the number of the agents
in a network. On the other hand, the size of the matrix appearing in
the proposed conditions is equal to $n$, which allows us to apply the
conditions for large-scale networks.}

This paper is organized as follows. After giving the notation used in
this paper, in Section~\ref{sec:SwitchModel} we give the model of
disease spread over randomly switched networks. Then
Section~\ref{sec:EdgeIndependent} gives sufficient conditions for
epidemics dying out over randomly switched networks having independent
edges. Examples are presented in Section~\ref{sec:degrees}. Finally in
Section~\ref{sec:weight} we extend the obtained results to the case
when networks are modeled by weighted graphs.

\subsection{Mathematical Preliminaries}\label{sec:preliminary}

We let $I_n$ denote the $n\times n$ identity matrix. The subscript~$n$
will be omitted when it is obvious from the context. A real vector $v$
is said to be {\it nonnegative}, written $v\geq 0$, if it has only
nonnegative entries. By $\onev_n$ and $\onev_{n, m}$ we denote the
$n$-vector and the $n\times m$ matrix whose entries are all one. A
square matrix is said to be {\it Metzler} if its off-diagonal entries
are nonnegative. We denote the Kronecker product of matrices~$A$ and
$B$ by $A\otimes B$.

The Euclidean norm of $x\in \mathbb{R}^n$ is denoted by $\norm{x}$.
Also we define the $1$-norm of $x$ by $\norm{x}_1 = \sum_{i=1}^n
\abs{x_i}$. The maximum real part of the eigenvalues of $A$ is denoted
by $\eta(A)$. We say that $A$ is Hurwitz stable if $\eta(A) < 0$. The
matrix measure~\cite{Desoer1972} of $A\in \mathbb{R}^{n\times n}$ is
defined by $\mu(A) = \lim_{h\to 0}{(\norm{I+Ah} - 1)}/{h}$, where
$\norm{\cdot}$ denotes the maximum singular value. When $A$ is
symmetric and thus has only real eigenvalues, its maximum eigenvalue
is denoted by~$\lambda_{\max}(A)$.

An undirected graph is a pair~$\mathcal G = (\mathcal V, \mathcal E)$,
where $\mathcal V$ is a finite and nonempty set and $\mathcal E$ is a
set consisting of distinct and unordered pairs~$(i, j)$ of $i, j\in
\mathcal V$. We call the elements of~$\mathcal V$ vertexes and those
of $\mathcal E$ edges. The adjacency matrix~$A_{\mathcal G}$ of
$\mathcal G$ is defined as the square $\{0, 1\}$\nobreakdash-matrix of
size~$\abs{\mathcal V}$ whose $(i,j)$ entry is one if and only if
\mbox{$(i,j) \in \mathcal E$}. Clearly the adjacency matrix of an
undirected graph is symmetric and has zero diagonals. We say that
$j\in \mathcal V$ is a neighbor of~$i\in \mathcal V$ (or $i$ and $j$
are adjacent) if $(i, j) \in \mathcal V$. The degree of $i$ is defined
as the number of the neighbors of~$i$. 
 
The expectation of a random variable is denoted by~$E[\cdot]$. All the
Markov processes appearing in this paper are assumed to have finite
state spaces. 

\section{Epidemic Model over Randomly Switched Networks}\label{sec:SwitchModel}

In this section we give a model for the spread of disease over
randomly switched networks using switched differential equations. The
model is based on a model called $N$-intertwined model originally
proposed for epidemics over static networks~\cite{VanMieghem2009a}.
Then we will review some results on the stability of switched linear
systems that will be used in this paper.

\subsection{$N$-intertwined Model for Static Networks}

We first give an overview of the $N$-intertwined
model~\cite{VanMieghem2009a} for so-called
susceptible-infected-susceptible types of disease. Let a population be
modeled by an undirected graph~$\mathcal G$ over vertexes $\{1,
\dotsc, n\}$. We regard each vertex as an agent who can be infected
and also can transmit disease to its neighbors. In the model, at each
time~$t\geq 0$, each vertex can be in one of the two states: {\it
susceptible} or {\it infected}. We assume that, when a vertex $i$ is
infected, the transition to the susceptible state occurs following a
Poisson process with rate~$\delta_i$, called {\it curing rate}. On the
other hand, if $i$ is in the susceptible state and one if its
neighbors is in the infected state, then $i$ will make a transition to
the infected state by following a Poisson process with rate~$\beta_j$,
called {\it infection rate}. All the Poisson processes are assumed to
be independent. Throughout this paper we assume that population has
{\it homogeneous infection rate $\beta$ and curing rate~$\delta$},
i.e., we assume that $\beta_1 = \cdots = \beta_n = \beta$ and
$\delta_1 = \cdots = \delta_n =\delta$.

Let $p_i(t)$ denote the probability that the vertex $i$ is infected.
Define $p = [p_1 \ \cdots \ p_n]^\top$ and $P = \diag(p_1, \dotsc,
p_n)$. Then the $N$-intertwined model~\cite{VanMieghem2009a} reads
\begin{equation}\label{eq:N-inter}
\frac{dp}{dt} = (\beta A_{\mathcal G}-\delta I)p - \beta P A_{\mathcal G}p,
\end{equation}
where $p(0) = p_0 \in [0, 1]^n$. We notice that clearly $p(t)\geq 0$
for every $t\geq 0$.

\subsection{Switched $N$-intertwined Model}\label{sec:switchedNinter}

Based on the $N$-intertwined model~\eqref{eq:N-inter} we can readily
state our model of disease spread over randomly switched networks. We
assume that our randomly switched network is modeled as $\mathcal
G_\sigma := \{\mathcal G_{\sigma(t)}\}_{t\geq 0}$, where $\sigma =
\{\sigma(t)\}_{t\geq 0}$ is a {\it time-homogeneous Markov process}
taking its values in $\{1, \dotsc, N\}$ and $\mathcal G_1, \dotsc,
\mathcal G_N$ are undirected graphs over vertexes $\{1, \dotsc, n\}$.
Examples of such random switched graphs include the activity driven
networks~\cite{Perra2012}, graphs with edge swapping~\cite{Volz2009},
and temporal exponential random graphs~\cite{Hanneke2010}. Then we
model the disease spread over the randomly switched  network $\mathcal
G_\sigma$ by the {\it switched $N$-intertwined model}
\begin{equation}\label{eq:def:Sigma}
\Sigma
\colon
\frac{dp}{dt} 
= 
(\beta A_{\mathcal G_{\sigma(t)}}-\delta I)p 
- 
\beta PA_{\mathcal G_{\sigma(t)}}p,
\end{equation}
where $p(0) = p_0 \in [0, 1]^n$ and $\sigma(0) = \sigma_0 \in \{1,
\dotsc ,N\}$ are arbitrary constants.

The principal aim of this paper is to give conditions under which the
zero equilibrium of $\Sigma$ is stable, i.e., the infection
probability $p$ converges or stays close to the origin. We introduce
the following definitions.

\begin{definition}\label{defn:stability}
We say that $\Sigma$ is 
\begin{enumerate}
\item {\it mean stable} if there exist $C>0$ and $\epsilon>0$ such
that, for every $p_0$ and $\sigma_0$, it holds that $E[\norm{p(t)}]
\leq Ce^{-\epsilon t} \norm{p_0}$;

\item {\it almost surely stable} if, for every $p_0$ and $\sigma_0$,
it holds that $P\left( \lim_{t\to\infty}\norm{p(t)} = 0 \right) = 1$.
\end{enumerate}
\afterequation
\end{definition}

The following linearized model
\begin{equation}\label{eq:def:barSigma}
\bar \Sigma\colon\frac{d\bar p}{dt} 
= 
(\beta A_{\mathcal G_{\sigma(t)}} - \delta I)\bar p
\end{equation}
of $\Sigma$ plays as an important role as the one for static networks
does in \cite{Preciado2013,Preciado2014}, where the authors propose the optimal
vaccination strategy for the disease spread over static networks. We
let the initial conditions of $\bar \Sigma$ given by
constants~\mbox{$\bar p(0) = \bar p_0 \in [0, 1]^n$} and $\sigma(0) =
\bar{\sigma}_0 \in \{1, \dotsc, N\}$. The next lemma is easy to see
but fundamental.

\begin{lemma}\label{lem:pleqbarp}
If $p_0 = \bar p_0$ and $\sigma_0 = \bar \sigma_0$, then we have
$\norm{p(t)}_1 \leq \norm{\bar p(t)}_1$ for every $t\geq 0$ with
probability one.
\end{lemma}

\begin{proof}
Let $e(t) = \bar p (t)- p(t)$. Then $e(0) =0$ and also
$\cramped{de/{dt} = \beta P A_{\mathcal G_{\sigma(t)}}p \geq 0}$
because all of $\beta$, $P$, $\cramped{A_{\mathcal G_{\sigma(t)}}\!}$,
and $p$ are nonnegative. Therefore $p(t) \leq \bar p(t)$ entry-wise
for every $t\geq 0$. Finally multiply $\onev_{n}^\top$ to the obtained
inequality from the left to obtain the desired inequality.
\end{proof}

Therefore, the study of the stability of $\Sigma$ can be reduced to
that of the switched linear system~$\bar \Sigma$, for which various
results from the theory of switched linear
systems~\cite{Costa2013} are available. We will review some of
the results in the next section.

\subsection{Stability of Markov Jump Linear Systems}

As in the previous section let $\sigma = \{\sigma(t)\}_{t\geq 0}$ be a
time-homogeneous Markov process taking its values in $\{1, \dotsc,
N\}$ and let $A_1, \dotsc, A_N \in \mathbb{R}^{n\times n}$. Consider
the switched linear system (called {\it Markov jump linear
system}~\cite{Costa2013})
\begin{equation}\label{eq:MJLS}
\frac{dx}{dt} = A_{\sigma(t)}x, 
\end{equation}
where $x(0) = x_0 \in \mathbb{R}^n$ and $\sigma(0) = \sigma_0$ are
arbitrary constants. The stability of Markov jump linear systems is
defined in the following standard way.

\begin{definition}\label{defn:cont:switch}
We say that \eqref{eq:MJLS} is
\begin{enumerate}
\item {\it mean stable} if there exist $C > 0$ and
$\epsilon>0$ such that $E[\norm{x(t)}] \leq Ce^{-\epsilon
t}\norm{x_0}$ for all $x_0$ and $\sigma_0$; 

\item {\it almost surely stable} if, for all $x_0$ and $\sigma_0$,
there holds $P\left(\lim_{t\to\infty} \norm{x(t)} = 0\right) = 1$.
\end{enumerate}
\afterequation
\end{definition}


We say that the Markov jump linear system~\eqref{eq:MJLS} is {\it
positive}~\cite{Ogura2013f} if $x_0 \geq 0$ implies
$x(t)\geq 0$ for every $t\geq 0$ with probability one. For the
system~\eqref{eq:MJLS} to be positive, it is necessary and sufficient
that all the matrices~$A_1$, $\dotsc$, $A_N$ are Metzler. For example,
$\bar \Sigma$ defined in \eqref{eq:def:barSigma} is a positive Markov
jump linear system. The next stability condition is obtained
in~\cite{Ogura2013f}.

\begin{proposition}\label{Prop:oguraBolzern}
Assume that the Markov jump linear system~\eqref{eq:MJLS} is positive.
Let $\Pi\in\mathbb{R}^{N\times N}$ be the infinitesimal generator
of $\sigma$. Then \eqref{eq:MJLS} is mean stable if and only if the
matrix
\begin{equation}\label{eq:stblmat:markov}
\mathcal A = \Pi^\top \otimes I_n + \diag(A_1, \dotsc, A_N)
\in \mathbb{R}^{nN\times nN}
\end{equation}
is Hurwitz stable.
\end{proposition}

We also recall the next proposition, which gives a sufficient
condition for the almost sure stability of not necessarily positive
Markov jump linear systems.

\begin{proposition}[{\cite[Theorem~4.2]{Fang2002c}}]\label{prop:as.stbl}
Assume that $\sigma$ has the unique stationary distribution~$\pi$.
Then \eqref{eq:MJLS} is almost surely stable if $E[\mu(A_\pi)] < 0$.
\end{proposition}

\section{Disease Spread over Randomly Switched Networks with Independent Edges}
\label{sec:EdgeIndependent}

The aim of this section is to give an easy to apply sufficient
condition for the stability of the switched $N$-intertwined
model~$\Sigma$ under the assumption that the underlying randomly
switched network has independent edges.

\subsection{Computational Difficulty}\label{sec:compdiff}

Before presenting the main result we first observe the
computational difficultly of the stability analysis. Theoretically,
the next proposition completely solves the problem by giving a
necessary and sufficient condition for mean stability via the
eigenvalues of a matrix.

\begin{proposition}\label{prop:meanstbl:Sigma}
Let $\Pi \in \mathbb{R}^{N\times N}$ be the infinitesimal generator of
the Markov process~$\sigma$. Define
\begin{equation}\label{eq:stblmat:markov:epidemic}
\mathcal A_\beta 
= 
\Pi^\top \otimes I_n 
+ 
\beta \diag(A_{\mathcal G_1}, \dotsc, A_{\mathcal G_N}).
\end{equation}
Then, $\Sigma$ is mean stable if and only if $\eta(\mathcal
A_\beta) < \delta$.
\end{proposition}

\begin{proof}
The Markov jump linear system~$\bar \Sigma$ is positive because the
matrix~$\beta A_{\mathcal G_i}-\delta I$ is Metzler for every $i$.
Therefore, using Proposition~\ref{Prop:oguraBolzern}, one can see that
$\eta(\mathcal A_\beta) < \delta$ implies that $\bar \Sigma$ is mean
stable. Thus, by Lemma~\ref{lem:pleqbarp}, we can show that $\Sigma$
is also mean stable.

On the other hand, assume that $\Sigma$ is mean stable. Here we
provide only a sketch of the proof. By Theorem~3.8 in
\cite{Khasminskii2007}, we can show that $\Sigma$ admits Lyapunov
functions. Then it turns out that these Lyapunov functions also apply
to $\bar \Sigma$ around a sufficiently small neighborhood of the
origin. This is because, in such a neighborhood, the second order
term~$\beta PA_{\mathcal G_\sigma}p$ of~$\Sigma$ can be ignored
compared with its linear term. Therefore we can conclude that $\bar
\Sigma$ is also mean stable. Hence, again by
Proposition~\ref{Prop:oguraBolzern}, we obtain $\eta(\mathcal A_\beta)
< \delta$.
\end{proof}

Though Proposition~\ref{prop:meanstbl:Sigma} gives a practical
characterization of stability when the size $n$ of a switched network
is effectively small, unfortunately, it cannot be easily used for real
world networks with a large $n$. The size of the matrix~$\mathcal
A_\beta$ equals $nN$, where $N$ is the number of all the possible
configurations of the switched network. This $N$ can be very large
when the network has many agents and can have various patterns of
topology. In the extreme case when all the $\binom{n}{2} = n(n-1)/2$
possible edges can be present or not independently, $N$ equals the
rapidly growing exponential~$2^{n(n-1)/2}$. By the same reason,
Proposition~\ref{prop:as.stbl} is often not easy to use in practice
when $n$ is large because it involves finding the matrix measure of a
potentially huge number of matrices. We also remark that, in such a
situation, it would be computationally too expensive to even estimate
a stability condition by simulation.

These difficulties motivate us to find easy to apply sufficient
conditions for stability when the size of a network is large.

\subsection{Stability Condition for Switched Edge-independent Networks} \label{sec:edge:main}

One of the classical models of large{-scale} and random but static
networks is Erd\H{o}s--R\'enyi graphs~\cite{ErdHos1959}, in which
edges are assigned uniformly and independently for each pair of
vertexes. Based on the fact that the distributions of the degrees of
Erd\H{o}s--R\'enyi graphs shows a large deviation from those of real
world networks, recently Chung~\cite{Chung2003a} proposed an improved
version of the graphs by removing the uniformity constraint. 

{Extending the
above models, in this paper we study the disease spread over
edge-independent random dynamical networks defined as follows.}

\begin{definition}\label{defn:independent.edges}
For distinct $i, j \in \{1, \dotsc, n\}$ we let \mbox{$A_{ij} =
\{A_{ij}(t)\}_{t\geq 0}$} denote the $(i,j)$-element of the
matrix-valued stochastic process~$A_{\mathcal G_\sigma}$. We say that
$\mathcal G_\sigma$ has {\it independent edges} if the processes
$\{A_{ij}\}_{i>j}$ are stochastically independent.
\end{definition}

Then our assumption on the randomly switched network~$\mathcal
G_\sigma$ can be stated as follows.

\begin{assumption}\label{assm:independent.edges}
$\mathcal G_\sigma$ has independent edges and $\sigma$ has a unique
stationary distribution.
\end{assumption}

\begin{remark}
The existence of a unique stationary distribution is not very
restrictive because a Markov process in general has a unique
stationary distribution under a mild assumption of irreducibility and
recurrence property.
\end{remark}

Let $\pi$ denote the stationary distribution of~$\sigma$. Then
$\mathcal G_{\pi}$ is the random stationary graph with the random
adjacency matrix $A_{\mathcal G_\pi}$. Define $\bar A = E[A_{\mathcal
G_\pi}]$ and
\begin{equation}\label{eq:def:Delta}
\Delta = \max_{1\leq i\leq n} \biggl(
\sum_{j=1}^n \bar A_{ij}(1 - \bar A_{ij})
\biggr). 
\end{equation}
{Without loss of generality we assume $\Delta > 0$, because
otherwise $\Delta = 0$ and therefore the graph process~$\mathcal
G_\sigma$ equals the static network having the $\{0, 1\}$-matrix $\bar
A$ as its adjacency matrix. In this case $\Sigma$ coincides with the
$N$-intertwined model~\eqref{eq:N-inter} for static networks and
therefore $\Sigma$ is stable if and only if \mbox{$\lambda_{\max}(\bar
A) < \delta /\beta$} as found in \cite{VanMieghem2009a}.}

The next theorem gives an easy-to-use alternative of
Proposition~\ref{prop:meanstbl:Sigma} and is the main result of this
paper.

\begin{theorem}\label{thm:independent.edges}
Define the function~$f$ on $[0, \infty)$ by
\begin{equation}\label{eq:def:f}
f(s) = s + 2n^2\exp\left(-\frac{3s^2}{2s+6\Delta}\right). 
\end{equation}
Then $\Sigma$ is almost surely stable if
\begin{equation}\label{eq:threshold:independent.edges}
\lambda_{\max}(\bar A) + \min_{s\geq 0}f(s) 
< 
{\delta}/{\beta}.
\end{equation}
\afterequation
\end{theorem}

\begin{remark}
{We can understand the quantity $\min_{s\geq 0}f(s)$ as the
measure of uncertainty because the quantity increases with respect to
$\Delta$, which measures the variance of the stationary random graph
$\mathcal G_\pi$.}
\end{remark}

{We note that the size of the matrix~$\bar A$ in the
condition~\eqref{eq:threshold:independent.edges} is only $n$ and is
much smaller than that of $\mathcal A_\beta$ in
Proposition~\ref{prop:meanstbl:Sigma}.} Also $\bar A$ can be found as
follows. Since $\mathcal G_\sigma$ has independent edges, each
scalar-valued stochastic process $A_{ij}$ is a $\{0, 1\}$-valued
time-homogeneous Markov process. Then its infinitesimal generator is
of the form
\begin{equation}\label{eq:infinitesimal.gen}
\begin{bmatrix}
-p_{ij}&p_{ij}\\q_{ij}&-q_{ij}
\end{bmatrix}
\end{equation}
for some $p_{ij}, q_{ij}\geq 0$. The uniqueness of the stationary
distribution of $\sigma$ yields that $p_{ij} + q_{ij} > 0$. Then we
can easily show that the stationary distribution of $A_{ij}$, denoted
by $\pi_{ij}$, is given by \mbox{$\pi_{ij}(\{0\}) = {q_{ij}}/{(p_{ij}
+ q_{ij})}$} and \mbox{$\pi_{ij}(\{1\}) = {p_{ij}}/{(p_{ij} +
q_{ij})}$}. Therefore we obtain \mbox{$\bar A_{ij} =
p_{ij}/(p_{ij}+q_{ij})$} when $i \neq j$ and $\bar A_{ii} = 0$ for
every~$i$.

Moreover, the next proposition shows that the minimum~$\min_{s\geq 0}f(s)$ in
\eqref{eq:threshold:independent.edges} can be found by solving a
convex program.

\begin{proposition}\label{prop:conv}
There exists $0 < s_0 < 2\Delta$ such that
$f$ is convex on $[s_0, \infty)$ and 
\begin{equation}\label{eq:min=convprog}
\min_{s\geq 0} f(s) = \min\Bigl({f(0), \min_{s\geq s_0} f(s)}\Bigr).
\end{equation}
\afterequation
\end{proposition}


\subsection{Proofs}

In this section we give the proofs of
Theorem~\ref{thm:independent.edges} and Proposition~\ref{prop:conv}.
For the proof of Theorem~\ref{thm:independent.edges} we state two
propositions. 

\begin{proposition}\label{prop:Sigma:a.s.stable}
If
$E[\lambda_{\max}(A_{\mathcal G_{\pi}})] < \delta/\beta$, 
then $\Sigma$ is almost surely stable.
\end{proposition}

\begin{proof}
Recall \cite{Desoer1972} that, if $A$ is a real symmetric matrix, then
$\mu(A) = \lambda_{\max}(A)$. Thus we have $\mu(A_{\mathcal G_\pi}) =
\lambda_{\max}(A_{\mathcal G_\pi})$ with probability one. Now assume
$E[\lambda_{\max}(A_{\mathcal G_{\pi}})] < \delta/\beta$. Then $E[\mu(\beta
A_{\mathcal G_\pi}-\delta I)] = E[\lambda_{\max}(\beta A_{\mathcal
G_\pi}-\delta I)]  < 0$. Therefore, by Proposition~\ref{prop:as.stbl},
$\bar \Sigma$ is almost surely stable; i.e., $\norm{\bar p(t)} \to 0$
as $t\to\infty$ with probability one. Therefore, by using the
equivalence of the norms $\norm{\cdot}$ and $\norm{\cdot}_1$ and also
Lemma~\ref{lem:pleqbarp}, we can show the almost sure stability of
$\Sigma$.
\end{proof}

To evaluate $E[\lambda_{\max}(A_{\mathcal G_{\pi}})]$ we will need the
following result from the spectral theory of random graphs.

\begin{proposition}[{\cite[p.~7]{Chung2011}}]\label{prop:ChungRad}
Let $\mathcal G$ be a random undirected graph on the vertex set~$\{1,
\dotsc, n \}$, where two vertexes are adjacent in $\mathcal G$
independently. Let $\bar A = E[A_{\mathcal G}]$ and define $\Delta$ by
\eqref{eq:def:Delta}. Then, for every $s\geq 0$, 
\begin{equation}\label{eq:ChungRad}
P(\{\lambda_{\max}(A) > \lambda_{\max}(\bar 
A) + s\})
\leq
2n \exp\left(-\frac{3s^2}{2s + 6\Delta}\right). 
\end{equation}
\afterequation
\end{proposition}

Now we prove Theorem~\ref{thm:independent.edges}. 

\begin{proofof}{Theorem~\ref{thm:independent.edges}}
Notice that the function $f$ is well defined. Since $\mathcal
G_\sigma$ has independent edges, any pair of vertexes are adjacent
independently with others in the random graph~$\mathcal G_\pi$.
Therefore Proposition~\ref{prop:ChungRad} applies to $\mathcal G_\pi$.
Let $\Omega$ denote the underlying probability space and define, for
each $s\geq 0$,
\begin{equation}
\Omega_s = \{\omega \in \Omega \colon
\lambda_{\max}(A_{\mathcal G_{\pi}}) > \lambda_{\max}(\bar A) + s \}.
\end{equation}
If \mbox{$\omega \in \Omega_s^c$}, then $\lambda_{\max}(A_{\mathcal
G_{\pi}}) \leq \lambda_{\max}(\bar A) + s$. On the other hand, if
$\omega \in \Omega_s$, then we have the trivial estimate
$\lambda_{\max}(A_{\mathcal G_{\pi}}) < n$. Therefore, by
\eqref{eq:ChungRad} we obtain
\begin{equation}\label{eq:hogen}
\begin{aligned}
E [\lambda_{\max}(A_{\mathcal G_{\pi}})] 
&<
P(\Omega_s^c) (\lambda_{\max}(\bar A) + s)
+
P(\Omega_s) n
\\
&\leq 
\lambda_{\max}(\bar A) + s + 2n^2 \exp\left(-\frac{3s^2}{2s + 6\Delta}\right)
\\
&=
\lambda_{\max}(\bar A) + f(s).
\end{aligned}
\end{equation}
Since $f$ is continuous and also $f(s)$ diverges to $+\infty$ as
\mbox{$s\to\infty$}, the minimum $\min_{s\geq 0} f(s)$ exists.
Therefore, taking the minimum with respect to $s\geq 0$ on the most
right hand side of \eqref{eq:hogen} proves \mbox{$E
[\lambda_{\max}(A_{\mathcal G_{\pi}})] \leq \lambda_{\max}(\bar A) +
\min_{s\geq 0}f(s)$}. Hence, if \eqref{eq:threshold:independent.edges}
holds, then $\Sigma$ is almost surely stable by
Proposition~\ref{prop:Sigma:a.s.stable}.
\end{proofof}

Then we give the proof of Proposition~\ref{prop:conv}.

\begin{proofof}{Proposition~\ref{prop:conv}}
Let us only prove the convexity part, as the
equation~\eqref{eq:min=convprog} follows from a straightforward
argument. Throughout this proof we shall work with the translation of
$f$ defined by $g(s) = f(s-3\Delta)$ ($s\geq 3\Delta$). We need to
show that there exists $3\Delta <s_0 < 5\Delta$ such that $g$ is
convex on $[s_0, \infty)$. Define $h_1, h_2\colon [3\Delta, \infty)
\to \mathbb{R}$ by
\begin{equation}
h_1(s) 
= 
\frac{c_1 (c_2 s^2 - c_3 + \sqrt{2c_3 s}) }{s^4\exp(c_2 s +c_3 s^{-1})}
,\ 
h_2(s) 
= 
c_2 s^2 - c_3 - \sqrt{2c_3s}, 
\end{equation} 
where $c_1 = 2n^2 e^{9\Delta}$, $c_2 = 3/2$, and $c_3 = 27\Delta^2/2$.
Then we can easily show that $g'' = h_1 h_2$. It is trivial to check
that $h_1 (s) \geq 0$. Also, since $h_2(3\Delta) < 0$ and $h_2'(s) >
0$, there exists $s_0 > 3\Delta$ such that $h_2 < 0$ on $[3\Delta,
s_0)$ and $h_2 > 0$ on $(s_0, \infty)$. Since \mbox{$h(5\Delta)  =
24\Delta^2 - 3\sqrt{15}\Delta > (24-3\sqrt{15})\Delta^2 > 0$}, we can
check $s_0 < 5\Delta$. The above argument yields $g''< 0$ on
$[3\Delta, s_0)$, as desired. Then we can prove
\eqref{eq:min=convprog} by carefully investigating the derivatives
$g'$ and $g''$. The details are omitted. 
\end{proofof}

\section{Examples}\label{sec:degrees}

In this section we apply Theorem~\ref{thm:independent.edges} to the
following important classes of dynamical graphs: randomly switched
graphs with communities~\cite{Fortunato2010} and expected
degrees~\cite{Chung2003a}.

\subsection{Switched Graphs with Community Structure}

Let us consider a population grouped into the two
communities~$\mathcal V_1 = \{1, \dotsc, n_1\}$ and $\mathcal V_2 =
\{n_1+1, \dotsc, n_1+n_2\}$. We assume that, for any pair of vertexes
$(i, j)$, the edge-process~$A_{ij}$ has the stationary
distribution~$\mu_{ij}$ that depends only on the communities to which
$i$ and $j$ belong. Therefore, there exist $\theta_1, \theta_2, \phi
\in [0, 1]$ such that
\begin{equation}
\mu_{ij}(\{1\})
=
\begin{cases}
\theta_\ell & i, j\in \mathcal V_\ell,\ \ell=1,2,\\
\phi & \text{otherwise}.
\end{cases}
\end{equation}
Then it follows that
\begin{equation}\label{eq:barA:groups}
\bar A = \begin{bmatrix}
\theta_1 \onev_{n_1, n_1} - \theta_1 I_{n_1} & \phi \onev_{n_1, n_2}
\\
\phi \onev_{n_2, n_1} & \theta_2 \onev_{n_2, n_2} - \theta_2 I_{n_2}
\end{bmatrix}.
\end{equation}
Therefore 
\begin{equation}
\begin{multlined}
\Delta = \max\bigl((n_1-1)\theta_1(1-\theta_1)+n_2\phi(1-\phi), \\
(n_2-1)\theta_2(1-\theta_2)+n_1\phi(1-\phi)\bigr). 
\end{multlined}
\end{equation} 
Also we can show that
\begin{equation}\label{eq:lammaxcomm}
\lambda_{\max}(\bar A) 
=
\frac{n_1\theta_1 + n_2\theta_2 
+ \sqrt{(n_1\theta_1 - n_2\theta_2)^2 + 4n_1n_2\phi^2}}{2} -\epsilon
\end{equation}
for some $\epsilon$ lying in between $\theta_1$ and $\theta_2$. 

For example let $n_1 = 10^4$, $n_2 = 10^5$, $\theta_1 = 0.5$,
$\theta_2 = 0.3$, and $\phi = 0.1$. Then one can compute
$\lambda_{\max}(\bar A) = 3.04\cdot 10^4$. Also, by solving the convex
program~\eqref{eq:min=convprog}, we can easily find $\max_{s\geq 0}
f(s) = 9.83\cdot 10^2$, which is effectively smaller than
$\lambda_{\max}(\bar A)$. Then, by
Theorem~\ref{thm:independent.edges}, $\Sigma$ is almost surely stable
if $3.14\cdot 10^4 < \delta/\beta$. We remark that, in this case, it
is almost impossible to use Proposition~\ref{prop:meanstbl:Sigma}
because the matrix~$\mathcal A_\beta$ has the dimension more than
$\cramped{10^{10^9}}$.

\subsection{Switched Graphs with Expected Degrees}

In this section we study the special case when the expected adjacency
matrix~$\bar A$ has the following structure.

\begin{assumption}\label{assm:alpha_i.alpha_j}
There exist $\alpha_1, \dotsc, \alpha_n \geq 0$ such that
\begin{equation}
\bar A_{ij} = \alpha_i \alpha_j
\end{equation}
for every distinct pair $(i,j)$.
\end{assumption}


Epidemiologically, one can regard the constant $\alpha_i$ as the
measure of the activity of the vertex $i$. Then we can understand
Assumption~\ref{assm:alpha_i.alpha_j} as stating that the frequency of
the interaction between two agents $i$ and $j$ is solely determined by
those activity measure.

Also Assumption~\ref{assm:alpha_i.alpha_j} can be supported by its
connection to one of the well known models of random graphs. Let $d\in
\mathbb{R}^n$ be nonzero and nonnegative and let $\rho =
1/(\sum_{i=1}^{n}d_i)$. We say that an undirected random
graph~$\mathcal G$ has {\it expected degrees $d$}~\cite{Chung2003a} if
edges are independently assigned to each pair of vertexes $(i, j)$
with probability $\rho d_id_j$. We can show the next proposition under
Assumptions~\ref{assm:independent.edges} and
\ref{assm:alpha_i.alpha_j}. The proof is straightforward and is hence
omitted.

\begin{proposition}
For each $i$ let $d_i = \alpha_i \sum_{j=1}^n \alpha_j$. The
stationary graph $\mathcal G_\pi$ has expected degrees~$d$.
\end{proposition}

The next theorem gives a sufficient condition for the almost sure
stability of $\Sigma$ in terms of the expected degrees of the
stationary graph~$\mathcal G_\pi$.

\begin{theorem}\label{thm:exp.degrees}
Let 
\begin{equation}\label{eq:m&tilde.d} 
\tilde d = \rho \sum_{i=1}^n d_i^2,
\ 
\Delta_d= \max_{1\leq i\leq n} \sum_{j=1}^n \rho
d_i d_j (\rho - d_i d_j)
\end{equation}
and define the function $f$ by \eqref{eq:def:f} with $\Delta$ replaced
by $\Delta_d$. Then $\Sigma$ is almost surely stable if
\begin{equation}\label{eq:cond:tilded}
\tilde d + \min_{s\geq 0}f(s) \leq {\delta}/{\beta}.
\end{equation}
Moreover~$\min_{s\geq 0}f(s)$
satisfies~\eqref{eq:min=convprog}.
\end{theorem}

\begin{proof}
It is straightforward to see that, for the random graph $\mathcal
G_\pi$, the quantity~$\Delta$ defined in \eqref{eq:def:Delta} equals
$\Delta_d$. Moreover, since $\bar A = \rho dd^\top$, we can show
$\lambda_{\max}(\bar A) = \tilde d$. Thus the sufficient condition
\eqref{eq:cond:tilded} immediately follows from
\eqref{eq:threshold:independent.edges}. The proof of the latter claim
is exactly the same as that of Proposition~\ref{prop:conv}.
\end{proof}

\begin{example}
One of the well-used model of the degree sequences is those with power
law distributions~\cite{Chung2003a,Chung2011}. We say that a degree
sequence~$d \in \mathbb{R}^n$ has the power law distribution with the
power law exponent $\beta > 2$, maximum degree $\Delta > 0$, and
average degree $\bar d > 0$ if
\begin{equation}
d_i = c (i+i_0-1)^{-{1}/{(\beta-1)}}
\end{equation} 
for every $i$, where 
\begin{equation}
c = \frac{\beta-2}{\beta-1}\bar dn^{1/(\beta-1)},\ i_0 = n\biggl(\frac{\bar d(\beta-2)}{m(\beta-1)}\biggr)^{\beta-1}.
\end{equation}
For example let $n = 10^7$, $\beta = 2.2$, $\Delta = 5 \cdot 10^5$,
and $\bar d = 10^3$. Then we have $\tilde d = 3.15\cdot 10^4$. Also,
by solving the convex program~\eqref{eq:min=convprog} we obtain
$\min_{s\geq 0} f(s) = 1.97\cdot 10^3$. Therefore, by
Theorem~\ref{thm:exp.degrees}, $\Sigma$ is almost surely stable if
\mbox{$3.35\cdot 10^4 < \delta/\beta$}. We remark that the quantity
$\min_{s\geq 0} f(s)$ measuring uncertainty is relatively small
compared with $\lambda_{\max}(\bar A)$.
\end{example}

\section{Disease Spread over Randomly Switched Networks with Weights}\label{sec:weight}

In this section, we extend the result obtained in
Section~\ref{sec:EdgeIndependent} to the networks modeled by weighted
graphs. We start by giving necessary definitions.

\begin{definition}
Let $\mathcal V$ be a finite set and let $\mathcal E$ be the set of
all the distinct and unordered pairs of the elements of~$\mathcal V$.
Also let $w \colon \mathcal E \to [0, \infty) \colon (i, j) \mapsto
w(i, j)$ be a function. We call the pair $(\mathcal V, w)$ a {
weighted undirected graph}. An element of~$\mathcal V$ ($\mathcal E$)
is called a vertex (edge, respectively). For an edge~$e\in\mathcal E$,
We call $w(e)$ the {weight} of $e$.
\end{definition}

Let $\mathcal V = \{1, \dotsc, n\}$. The adjacency
matrix~$A_{(\mathcal V, w)}\in \mathbb{R}^{n\times n}$ of an weighted
undirected graph $(\mathcal V, w)$ is defined by $[A_{(\mathcal V,
w)}]_{ij} = w(i, j)$ for $i\neq j$ and $[A_{(\mathcal V, w)}]_{ii} =
0$ for every $i$. When no confusion arises we write $A_{(\mathcal V,
w)}$ as $A_w$. We notice that $A_w$ is symmetric because $(i, j)$ and
$(j, i)$ are the same unordered pairs. Then we call the differential
equation
\begin{equation}\label{eq:N-inter:weighted}
\frac{dp}{dt} = (\beta A_w -\delta I) p - \beta PA_wp
\end{equation}
the $N$-intertwined model of disease spread over $(\mathcal V, w)$.
Epidemiologically, the weight $w(i, j)$ expresses the strength of the
communication between the agents $i$ and $j$, in the sense that
instantaneous rate that disease transmits from~$i$ to~$j$ (provided
$i$ is infected and $j$ is not) equals the product~$\beta w(i, j)$.
This in particular implies that, without loss of generality, we can
normalize $w$ as
\begin{equation}\label{eq:normalize.w}
w(i, j) \leq 1
\end{equation}
by taking $\beta$ sufficiently large.

Then, as in Section~\ref{sec:switchedNinter}, let us assume that the
weight of the given network of agents changes over time according to a
time-homogeneous Markov process, i.e., suppose that there exists a
time-homogeneous Markov process $\sigma = \{\sigma(t)\}_{t\geq 0}$
such that the weight of the graph at time $t\geq 0$ is given by the
function $w_{\sigma(t)} \colon \mathcal E \to [0, \infty) \colon (i,
j) \mapsto w_{\sigma(t)}(i, j)$. Then we can model the disease spread
over time-varying, weighted, and undirected graph $\mathcal G_\sigma
:= (\mathcal V, w_\sigma)$ by the switched differential equation
\begin{equation}
\Sigma_w 
\colon 
\frac{dp}{dt} = (\beta A_{w_{\sigma(t)}} -\delta I) p - \beta PA_{w_{\sigma(t)}}p. 
\end{equation}
The almost sure stability of $\Sigma_w$ is defined in the same way as
Definition~\ref{defn:stability}. Also, extending
Definition~\ref{defn:independent.edges}, We say that $\mathcal
G_\sigma$ has {\it independent edges} if the stochastic
processes~$\{w_{\sigma(\cdot)}(i, j)\}_{(i, j)\in \mathcal E}$ are
independent.

If $\sigma$ has a stationary distribution~$\pi$, then we let
$\Var(A_{w_\pi})$ be the $n\times n$ real matrix obtained by taking
the variances of the random matrix~$A_{w_\pi}$ entry-wise. Since
$\sigma$ is assumed to have finitely many states, each entry
of~$A_{w_\pi}$ is a distribution having finite support so that it is
straightforward to find $\Var(A_{w_\pi})$. The next theorem extends
Theorem~\ref{thm:independent.edges} to the weighted and randomly
switched networks.

\begin{theorem}\label{thm:independent.edges:weighted}
Assume that $\mathcal G_\sigma$ has independent edges and $\sigma$ has
the unique stationary distribution~$\pi$. Let $\bar A = E[A_{w_\pi}]$
and also let $\varDelta$ be the maximum row sum of $\Var(A_{w_\pi})$.
Define the function~$f$ by \eqref{eq:def:f}. If
\eqref{eq:threshold:independent.edges} holds, then $\Sigma_w$ is
almost surely stable. Moreover the minimum $\min_{s\geq 0}f(s)$ can be
found by solving the convex program~\eqref{eq:min=convprog}.
\end{theorem}

\begin{proof}
From~\cite{Chung2011} we can easily check that
Proposition~\ref{prop:ChungRad} holds for a weighted random graph as
long as its weights are all less than one with probability one and we
replace $\Delta$ by the maximum row sum of $\Var(A_{\mathcal G})$.
Therefore, under the normalization~\eqref{eq:normalize.w}, we can
apply Proposition~\ref{prop:ChungRad} to the stationary random
graph~$\mathcal G_\pi$ and prove the theorem in the same way as
Theorem~\ref{thm:independent.edges}. The details of the proof are
omitted.
\end{proof}

\section{Conclusion}

We studied the disease spread over randomly switched networks having
stochastically independent edges. The disease spread was modeled by a
switched version of the $N$\nobreakdash-intertwined model for static
networks. Using the stability theory of Markov jump linear systems and
the spectral theory of random matrices, we gave sufficient conditions
for epidemics dying out almost surely. {We can check the proposed
conditions by finding the maximum real eigenvalue of a matrix whose
size equal the number of the agents and thus is highly efficient
compared with another condition based solely on the stability theory.}
We also gave an extension to the case when networks are modeled by
weighted graphs.



\end{document}